\documentclass[11pt,twoside]{article}
\usepackage[utf8]{inputenc}
\usepackage{jeffe} 
\usepackage[hmargin=1in,vmargin=1in]{geometry}

\usepackage[nocompress]{cite}
\usepackage{amsmath,mathtools,stmaryrd}

\usepackage{verbatim}

\def\etal{\textit{et~al.}}
\def\polylog{\mathop{\mathrm{polylog}}}
\def\eps{\varepsilon}

\def\seq#1{\langle #1 \rangle}

\def\Ceil#1{\left\lceil #1 \right\rceil}
\def\Seq#1{\left\langle #1 \right\rangle}
\def\Set#1{\left\{ #1 \right\}}

\def\Brack#1{\left[ #1 \right]}		

\def\arcto{\mathord\shortrightarrow}

\newtheorem{lemma}{Lemma}[section]
\newtheorem{theorem}[lemma]{Theorem}
\newtheorem{corollary}[lemma]{Corollary}

\numberwithin{figure}{section}

\def\corrs{\Pi}

\def\reparamP{\sigma}
\def\reparamQ{\theta}
\def\dist{\mathtt{d}}
\def\cost{\mu}

\def\frechet{\mathsf{FD}}

\begin{document}
\begin{titlepage}

\title{Approximating the (Continuous) Fr\'{e}chet Distance\thanks{%
Most of this work was done while the first author was a student at the University of Texas at
Dallas.}}

\author{Connor Colombe\thanks{%
    The University of Texas at Austin;
    \url{ccolombe@utexas.edu}}
		\and
		Kyle Fox\thanks{%
    The University of Texas at Dallas;
    \url{kyle.fox@utdallas.edu}}
}

\maketitle

\begin{abstract}
  We describe the first strongly subquadratic time algorithm with subexponential approximation ratio
  for approximately computing the Fr\'echet distance between two polygonal chains.
  Specifically, let~\(P\) and~\(Q\) be two polygonal chains with~\(n\) vertices in \(d\)-dimensional
  Euclidean space, and let \(\alpha \in [\sqrt{n}, n]\).
  Our algorithm deterministically finds an \(O(\alpha)\)-approximate Fr\'echet correspondence in
  time \(O((n^3 / \alpha^2) \log n)\).
  In particular, we get an \(O(n)\)-approximation in near-linear \(O(n \log n)\) time, a vast
  improvement over the previously best know result, a linear time \(2^{O(n)}\)-approximation.
  As part of our algorithm, we also describe how to turn any approximate decision
  procedure for the Fr\'echet distance into an approximate optimization algorithm whose
  approximation ratio is the same up to arbitrarily small constant factors.
  The transformation into an approximate optimization algorithm increases the running time of the
  decision procedure by only an \(O(\log n)\) factor.
\end{abstract}

\setcounter{page}{0}
\thispagestyle{empty}
\end{titlepage}

\pagestyle{myheadings}
\markboth{Connor Colombe and Kyle Fox}{Approximating the (Continuous) Fr\'{e}chet Distance}


\section{Introduction}
\label{sec:introduction}

The Fr\'echet distance is a commonly used method of measuring the similarity between a pair of
curves.
Both its standard (continuous) and discrete variants have seen use in map construction and
mapping~\cite{akpw-mca-15,cdgnw-ammrf-11}, handwriting recognition~\cite{skb-fdbas-07}, and protein
alignment~\cite{jxz-psadf-08}.

Formally, it is defined as follows:
Let \(P : [1, m] \to \R^d\) and \(Q : [1, n] \to \R^d\) be two curves in \(d\)-dimensional Euclidean
space.
We'll assume \(P\) and \(Q\) are represented as \EMPH{polygonal chains}, meaning there exist ordered
\EMPH{vertex} sequences \(\seq{p_1, \dots, p_m}\) and \(\seq{q_1, \dots, q_n}\) such that \(P(i) =
p_i\) for all \(1 \leq i \leq m\), \(Q(j) = q_j\) for all \(1 \leq j \leq n\), and both \(P\) and
\(Q\) are linearly parameterized along line segments or \EMPH{edges} between these positions.
We define a \EMPH{re-parameterization} \(\reparamP : [0, 1] \to [1, m]\) of \(P\) as any continuous,
non-decreasing function such that \(\reparamP(0) = 1\) and \(\reparamP(1) = m\).\footnote{%
Re-parameterizations are normally required to be bijective, but we relax this requirement to
simplify definitions and arguments throughout the paper.}
We define a re-parameterization \(\reparamQ : [0, 1] \to [1, n]\) of \(Q\) similarly.
We define a \EMPH{Fr\'echet correspondence} between \(P\) and \(Q\) as a pair \((\reparamP,
\reparamQ)\) of re-parameterizations of \(P\) and \(Q\) respectively, and we say any pair of reals
\((\reparamP(r), \reparamQ(r))\) for any \(0 \leq r \leq 1\) are \EMPH{matched} by the
correspondence.
Let \(\dist(p, q)\) denote the Euclidean distance between points \(p\) and \(q\) in \(\R^d\).
The \EMPH{cost} of the correspondence is defined as
\[\cost((\reparamP, \reparamQ)) := \max_{0 \leq r \leq 1} \dist(P(\reparamP(r)), Q(\reparamQ(r))).\]
Let \(\corrs_{\frechet}\) denote the set of all Fr\'echet correspondences between \(P\) and \(Q\).
The \EMPH{(continuous) Fr\'echet distance} of \(P\) and \(Q\) is defined as
\[ \frechet(P, Q) := \min_{(\reparamP, \reparamQ) \in \corrs_{\frechet}} \cost((\reparamP,
\reparamQ)).\]

The standard intuition given for this definition is to imagine a person and their dog walking along
\(P\) and \(Q\), respectively, without backtracking.
The person must keep the dog on a leash, and the goal is to pace their walks as to minimize the
length of leash needed to keep them connected.
There also exists a variant of the distance called the \EMPH{discrete Fr\'echet distance} where the
input consists of two finite \emph{point sequences}.
Here, we replace the person and dog by two frogs.
Starting with both frogs on the first point of their sequences, we must iteratively move the first,
the second, or both frogs to the next point in their sequences.
As before, the goal is to minimize the maximum distance between the frogs.

Throughout this paper, we assume \(2 \leq m \leq n\).
We can easily compute the \emph{discrete} Fr\'echet distance in~\(O(mn)\) time using dynamic
programming.
The first polynomial time algorithm for computing the continuous case was described by Alt and
Godau~\cite{ag-cfdbt-95}.
They use parametric search~\cite{c-sdsno-87,m-apcad-83} and a quadratic time decision procedure (see
Section~\ref{sec:preliminaries}) to compute the Fr\'echet distance in~\(O(mn \log n)\) time.
Almost two decades passed before Agarwal \etal~\cite{aaks-cdfds-14} improved the running time for
the discrete case to~\(O(mn \log \log n / \log n)\).
Buchin \etal~\cite{bbmm-fswdi-17} later improved the running time for the continuous case to~\(O(mn
(\log \log n)^2)\) (these latter two results assume we are working in the word RAM model of
computation).

Recently, Gudmundsson \etal~\cite{gmmw-ffdcl-19} described an~\(O(n \log n)\) time algorithm for
computing the continuous distance between chains~\(P\) and~\(Q\) assuming all edges have length a
sufficiently large constant larger than~\(\frechet(P, Q)\).
In short, having long edges allows one to greedily move the person and dog along their respective
chains while keeping their leash length optimal.

From this brief history, one may assume substantially faster algorithms  are finally forthcoming for
general cases of the continuous and discrete Fr\'echet distance.
Unfortunately, more meaningful improvements may not be possible;
Bringmann~\cite{b-wwdtt-14} showed that \EMPH{strongly subquadratic}
(\(n^{2-\Omega(1)}\)) time algorithms would violate the \emph{Strong Exponential Time
Hypothesis} (SETH) that solving CNF-SAT over \(n\) variables
requires~\(2^{(1-o(1))n}\) time~\cite{ip-ock-01}.

Therefore, we are motivated to look for fast \emph{approximation algorithms} for these problems.
Aronov \etal~\cite{ahkww-fdcr-06} described a~\((1+\eps)\)-approximation algorithm for the discrete
Fr\'echet distance.
This algorithm runs in subquadratic and often near-linear time if \(P\) or \(Q\) fall into one of a
few different ``realistic'' families of curves
such as ones modeling protein backbones.
Driemel \etal~\cite{dhw-afdrc-12} describe a \((1 + \eps)\)-approximation for the continuous
Fr\'echet distance that again runs more quickly if one of the curves belongs to a realistic family
than it would otherwise.
This latter algorithm was improved for some cases by Bringmann and K\"{u}nnemann~\cite{bk-iafdc-17}.
In the same work mentioned above, Gudmundsson \etal~\cite{gmmw-ffdcl-19} described a
\(\sqrt{d}\)-approximation algorithm that runs in linear time if the input polygonal chains have
sufficiently long edges.

Approximation appears more difficult when the input is arbitrary.
Bringmann~\cite{b-wwdtt-14} showed there is no strongly subquadratic time \(1.001\)-approximation
for the Fr\'echet distance, assuming SETH.
For arbitrary point sequences, Bringmann and Mulzer~\cite{bm-adfd-16} described an
\(O(\alpha)\)-approximation algorithm for the discrete distance for any \(\alpha \in [1, n / \log
n]\) that runs in \(O(n \log n + n^2 / \alpha)\) time.
Chan and Rahmati~\cite{cr-iaadf-18} later described an \(O(n \log n + n^2 / \alpha^2)\) time
\(O(\alpha)\)-approximation algorithm for the discrete distance for any \(\alpha \in [1, \sqrt{n /
\log n}]\).

For the \emph{continuous} Fr\'echet distance over arbitrary polygonal chains, the only strongly
subquadratic time algorithm known with bounded approximation ratio runs in linear time but has an
\emph{exponential} worst case approximation ratio of \(2^{\Theta(n)}\).
This result is described in the same paper of Bringmann and Mulzer~\cite{bm-adfd-16} mentioned
above.
We note that there is also a substantial body of work on the (approximate) nearest neighbor problem
using Fr\'echet distance as the metric;
see Mirzanezhad~\cite{m-annqc-20} for a survey of recent results.
These results assume the query curve or the curves being searched are short, so they do not appear
directly useful in approximating the Fr\'echet distance between two curves of arbitrary complexity.

The closely related problems of computing the dynamic time warping and geometric edit distances have
a similar history to that of the discrete Fr\'echet distance.\footnote{%
The dynamic time warping distance is defined similarly to the discrete Fr\'echet distance, except
the goal is to minimize the \emph{sum} of distances between the frogs over all pairs of points they
stand upon.
The geometric edit distance can be defined as the minimum number of point insertions and deletions
plus the minimum total cost of point substitutions needed to transform one input sequence into
another.
The cost of a substitution is the distance between its points.}
They have straightforward quadratic time dynamic programming algorithms that have been improved by
(sub-)polylogarithmic factors for some low dimensional cases~\cite{gs-dtwge-18};
substantial improvements to these algorithms violate SETH or other complexity theoretic
assumptions~\cite{bi-edccs-18,abw-thrlo-15,bk-qclbs-15,ahww-sbped-16};
and there are fast \((1 + \eps)\)-approximation algorithms specialized for realistic input
sequences~\cite{afpy-adtwe-16,ypfa-seaad-16}.
And, there exist some approximation results for arbitrary point sequences as well.
Kuszmaul~\cite{k-dtwss-19} described \(O((n^2 / \alpha) \polylog n)\) time
\(O(\alpha)\)-approximation algorithms for dynamic time warping distance over point sequences in
\emph{well separated tree metrics} of exponential spread and geometric edit distance over point
sequences in arbitrary metrics.
Fox and Li~\cite{fl-aged-19} described a randomized \(O(n \log^2 n + (n^2 / \alpha^2) \log n)\) time
\(O(\alpha)\)-approximation algorithm for geometric edit distance for points in low dimensional
Euclidean space.
Even better approximation algorithms exist for the traditional string edit distance where all
substitutions have cost exactly~\(1\);
see, for example, Andoni and Nosatzki~\cite{an-edntc-20}.

Each of the above problems for point \emph{sequences} admit strongly subquadratic approximation
algorithms with polynomial approximation ratios when the input comes from low dimensional Euclidean
space.
However, such a result remains conspicuously absent for the continuous Fr\'echet distance over
arbitrary polygonal chains.
One may naturally assume results for the discrete Fr\'echet distance
extend to the continuous case.
However, one advantage of discrete Fr\'echet distance over the continuous case is that input points
can only be matched with other input points.
The fact that vertices can match with edge interiors in the continuous case makes it much more
difficult to make approximately optimal decisions.
In addition, we can no longer depend upon certain data
structures for testing equality of subsequences in constant time.
These structures are largely responsible for the relatively small running times seen in the
algorithms of Chan and Rahmati~\cite{cr-iaadf-18} and Fox and Li~\cite{fl-aged-19}.

\subsection*{Our results}

We describe the first strongly subquadratic time algorithm with subexponential approximation ratio
for computing Fr\'echet correspondences between polygonal chains.
Let~\(P\) and~\(Q\) be two polygonal chains of~\(m\) and~\(n\) vertices, respectively,
in \(d\)-dimensional Euclidean space, and let \(\alpha \in [\sqrt{n}, n]\).
Again, we assume \(m \leq n\).
Our algorithm deterministically finds a Fr\'echet correspondence between~\(P\) and~\(Q\) of cost
\(O(\alpha)\cdot\frechet(P, Q)\) in time \(O((n^3 / \alpha^2) \log n)\).
In particular, we get an \(O(n)\)-approximation in near-linear \(O(n \log n)\) time, a vast
improvement over Bringmann and Mulzer's~\cite{bm-adfd-16} linear time \(2^{O(n)}\)-approximation for
continuous Fr\'echet distance.
Our algorithm employs a novel combination of ideas from the original exact algorithm of Alt and
Godau~\cite{ag-cfdbt-95} for continuous Fr\'echet distance, the algorithm of Chan and
Rahmati~\cite{cr-iaadf-18} for approximating the discrete Fr\'echet distance, and Gudmundsson
\etal's~\cite{gmmw-ffdcl-19} greedy approach for computing the Fr\'echet distance between chains
with long edges.

Let \(\delta > 0\).
We describe an \EMPH{approximate decision procedure} that either determines \(\frechet(P, Q) >
\delta\) or finds a Fr\'echet correspondence of cost \(O(\alpha) \cdot \delta\).
The \emph{exact} decision procedure of Alt and Godau~\cite{ag-cfdbt-95} computes a set of
\emph{reachability intervals} in the \emph{free space diagram} of \(P\) and \(Q\) with respect to
\(\delta\) (see Section~\ref{sec:preliminaries}).
These intervals represent all points on a single edge of \(Q\) that can be matched to a vertex of
\(P\) (or vice versa) in a Fr\'echet correspondence of cost at most \(\delta\).
For our approximate decision procedure, we compute a set of \emph{approximate reachability
intervals} such that the re-parameterizations realizing these intervals have cost \(O(\alpha) \cdot
\delta\).
We cannot afford to compute intervals for all \(\Theta(mn)\) vertex-edge pairs, so we instead focus
on \(O(n^2 / \alpha^2)\) vertex-edge pairs as described below that contain the first and last
vertices and edges of both chains.
The approximate interval we compute for any vertex-edge pair contains the exact interval for that
same pair.
So if \(\frechet(P, Q) \leq \delta\), we are guaranteed \((p_m, q_n)\) is approximately reachable
and our desired Fr\'echet correspondence exists.

The vertex-edge pairs chosen to hold the approximate reachability intervals follow from ideas of
Chan and Rahmati~\cite{cr-iaadf-18}.
Similar to them, we place a grid of side length \(\alpha \cdot \delta\) so that at most
\(O(n / \alpha)\) vertices of \(P\) and \(Q\) lie within distance \(3\delta\) of the side of a grid
box.  We call these \(O(n / \alpha)\) vertices \emph{bad} and the rest \emph{good}.
Also, we call any edge with a bad endpoint bad.
Our approximate reachability intervals involve only bad edges and vertices with at least one bad
incident edge.
To compute these intervals, we describe a method for tracing how a Fr\'echet correspondence of cost
\(\delta\) must behave starting from one approximate reachability interval until it reaches some
others we wish to compute.
Recall, an approximate reachability interval corresponds to pairs of points on \(P\) and \(Q\) that
could be matched together.
Either the next edge of \(P\) or \(Q\) after one of these pairs to leave a box is good and therefore
long, or it is bad, and we can afford to compute some new approximate reachability intervals using
this edge.
We can easily compute correspondences between long edges and arbitrary length edges on the other
curve, and we can greedily match the portions of the curves before they leave the box at cost at
most \(O(\alpha)\cdot \delta\).
The traces take only \(O(n)\) time each, and we perform at most \(O(n^2 / \alpha^2)\) traces, so our
decision procedure takes \(O(n^3 / \alpha^2)\) time total.

We would like to use our approximate decision procedure as a black box to compute a Fr\'echet
correspondence of cost \(O(\alpha) \cdot \frechet(P, Q)\) without knowing \(\frechet(P, Q)\) in
advance.
Unfortunately, we are unaware of any known general method to do so.\footnote{%
Bringmann and K\"{u}nnemann~\cite[Lemma 2.1]{bk-iafdc-17} claim there exists a general method for
turning an approximate decision procedure into an approximate optimization algorithm when
the approximation ratio of the decision procedure is at most \(2\).
However, they rely on a method of Driemel \etal~\cite{dhw-afdrc-12} that uses certain structural
properties of the input polygonal chains that we cannot assume.}
Therefore, we describe how to turn any approximate decision procedure into an algorithm with the
same approximation ratio up to arbitrarily small constant factors after an~\(O(\log n)\) factor
increase in running time.
In particular, any improvement to our approximate decision procedure would immediately carry over to
our overall approximation algorithm.
Our method involves binary searching over a set of \(O(n)\) values approximating distances between
pairs of vertices.
If there is a large gap between the Fr\'echet distance and the nearest of these values, we can
\emph{simplify} both \(P\) and \(Q\) without losing much accuracy in the Fr\'echet distance
computation while allowing us to use the long edge exact algorithm of Gudmundsson
\etal~\cite{gmmw-ffdcl-19}.

The rest of our paper is organized as follows.
We describe preliminary notions in Section~\ref{sec:preliminaries}.
We describe our decision procedure in Section~\ref{sec:decision_procedure} and how to turn it into
an approximation algorithm in Section~\ref{sec:approximation_algorithm}.
We conclude with some closing thoughts in Section~\ref{sec:conclusion}.

\section{Preliminaries}
\label{sec:preliminaries}

Let \(R : [1, n] \to \R^d\) be a polygonal chain in \(d\)-dimensional Euclidean space.
We let \(R[r, r']\) denote the restriction of \(R\) to \([r, r']\).
In other words, the notation refers to the portion of \(R\) between points \(R(r)\) and \(R(r')\).
We generally use \(s\) to refer to members of the domain of a polygonal chain \(P\) and \(t\) to
refer to members of the domain of a polygonal chain \(Q\).
We use \(i\) and \(j\), respectively, when these members are integers.
Recall, \(p_i = P(i)\) for all \(1 \leq i \leq m\) and \(q_j = Q(j)\) for all \(1 \leq j \leq n\).
We use superscript notation (\(s^a\)) to label particular members of these domains (and \emph{not}
to take the \(a\)th power of \(s\)), and we use subscript notation (\(s_k\)) when we are working
with an ordered list of these members.


\paragraph*{Free space diagram and reachability}

Let \(P : [1, m] \to \R^d\) and \(Q : [1, n] \to \R^d\) be polygonal chains.
Fix some \(\delta > 0\).
Alt and Godau~\cite{ag-cfdbt-95} introduced the \EMPH{free space diagram} to decide if \(\frechet(P,
Q) \leq \delta\).
It consists of a set of pairs \(F = \Set{(s, t) \in [1, m] \times [1, n]}\).
Each \((s, t) \in F\) represents the pair of points \(P(s)\) and \(Q(t)\).
Point \((s, t) \in F\) is \EMPH{free} if \(\dist(P(s), Q(t)) \leq \delta\).
The \EMPH{free space} \(\mathcal{D}_{\leq \delta}(P,Q)\) consists of all free points between $P$ and
$Q$ for a given $\delta$.
Formally, it is given by the set \(\mathcal{D}_{\leq \delta}(P,Q) := \{(s,t) \in [1,m] \times [1,n]:
\dist(P(s), Q(t)) \leq \delta\}\).
We say that a point \((s', t') \in F\) is \EMPH{reachable} if there exists an \(s\) and
\(t\)-monotone path from $(1,1)$ to $(s',t')$ through \(\mathcal{D}_{\leq \delta}(P, Q)\).

The standard procedure for determining if \(\frechet(P, Q) \leq \delta\) divides \(F\) into cells
\(C_{i,j} := [i - 1, i] \times [j - 1, j]\) for all \(i \in \Seq{2, \dots, m}\) and \(j \in \Seq{2,
\dots, n}\).
The intersection of a cell \(C_{i,j}\) with the free space is convex~\cite{ag-cfdbt-95}.
The intersection of an edge of the free space diagram cell $C_{i,j}$ with the free space forms a \EMPH{free space interval}.
The subset of reachable points within a free space interval form what is called an (exact)
\EMPH{reachability interval}.
We say a Fr\'echet correspondence \((\reparamP, \reparamQ)\) between \(P\) and \(Q\) \EMPH{uses} or
\EMPH{passes through} a reachability interval if there exists some point \((\reparamP(r),
\reparamQ(r))\) within that interval.

Given the bottom and left reachability intervals of a free space diagram cell, we can compute the
top and right reachability intervals of the same cell in $O(1)$ time~\cite{ag-cfdbt-95}.
The exact decision procedure loops through the cells in increasing order of \(i\) and \(j\),
computing reachability intervals one-by-one.
Let \(\alpha \in [\sqrt{n}, n]\).
We cannot afford to compute all \(\Theta(mn)\) reachability intervals, so instead we compute \(O(n^2
/ \alpha^2)\) \EMPH{(\(\alpha\))-approximate reachability intervals}.
The approximate reachability intervals are subsets of the free space intervals such that for any
point \((s, t)\) on an approximate reachability interval, there exists a Fr\'echet correspondence
between \(P[1, s]\) and \(Q[1, t]\) of cost \(O(\alpha)\cdot \delta\).
We express exact or approximate reachability intervals by the subset of \(F\) they contain;
for example, given \(j - 1 \leq t^a \leq t^b \leq j\), we will use \(\Set{i} \times [t^a, t^b]\) to
refer to an interval on the right side of cell \(C_{i,j}\).

\paragraph*{Grids, good points, bad points, and dangerous points}
 Chan and Rahmati~\cite{cr-iaadf-18} utilize a $d$-dimensional grid to create the useful notion of
 good and bad vertices for their discrete Fr\'echet distance approximation algorithm.
 We adopt their use of a $d$-dimensional grid.
 Unlike Chan and Rahmati, however, we are no longer working with sequences of discrete points but
 instead polygonal chains.
 We must therefore define new constructs of good and bad that work better for our problem's input.

 Let \(P : [1, m] \to \R^d\) and \(Q : [1, n] \to \R^d\) be two polygonal chains.
 Fix \(\delta > 0\) and \(\alpha \in [\sqrt{n}, n]\).
 Let $G$ be a $d$-dimensional grid consisting of \EMPH{boxes} of side length $\alpha \cdot \delta$.
 (We do not use the term \emph{cell} here to avoid confusion with the free space diagram.)
 We say a vertex of \(P\) or \(Q\) is \EMPH{good} if it is more than distance $3\delta$ from any
 edge of $G$.
 If a vertex is not good, then we call it \EMPH{bad}.
 For simplicity, we also designate \(p_1\), \(q_1\), \(p_m\), and \(q_n\) as bad, regardless of
 their position within boxes of \(G\).

We also extend the constructs of good and bad to the edges of \(P\) and \(Q\).
We say an edge on either chain is \EMPH{good} if both its endpoints are good vertices.
Otherwise, the edge is \EMPH{bad}.
Lastly, we say that a vertex is \EMPH{dangerous} (but not necessarily good or bad) if at least one of its incident edges is bad.
Chan and Rahmati~\cite[Lemma 1]{cr-iaadf-18} demonstrate how to compute a grid \(G\) with \(O(n /
\alpha)\) bad vertices in \(O(n)\) time.
Because each bad vertex has up to two incident edges, there are also $O(n/\alpha)$ bad edges.
Each bad edge is incident to two vertices, so there are $O(n/\alpha)$ dangerous vertices as well.
Our approximate decision procedure will compute approximate reachability intervals only between
dangerous vertices and bad edges.
Therefore, there will be at most \(O(n^2 / \alpha^2)\) such intervals.

\paragraph*{Curve simplification}
Let \(R : [1, n] \to \R^d\) be a polygonal chain with vertices \(\Seq{r_1, \dots, r_n}\).
Our approximation algorithm relies on a method for simplifying chains so their edges are not too
short.
We slightly modify of a procedure of Driemel~\etal~\cite{dhw-afdrc-12}.
Let \(\nu > 0\) be a parameter.
We mark \(r_1\) and set it as the \emph{current vertex}.
We then repeat the following procedure until we no longer have a designated current vertex.
We scan \(R\) from the current vertex until reaching the first vertex \(r_i\) of distance at least
\(\nu\) from the current vertex.
We mark \(r_i\), set it as the current vertex, and perform the next iteration of the loop.
The \EMPH{\(\nu\)-simplification} of \(R\), denoted \(\hat{R}\), is the polygonal chain consisting of
exactly the marked vertices in order.
Note that unlike Driemel~\etal~\cite{dhw-afdrc-12}, we do not require the final vertex of \(R\) to
be marked.
We can easily verify that all edges of \(\hat{R}\) have length at least \(\nu\).
Also, \(\frechet(R, \hat{R}) \leq \nu\)~\cite[Lemma 2.3]{dhw-afdrc-12}.

\section{Approximate Decision Procedure}
\label{sec:decision_procedure}

In this section, we present our \(O(\alpha)\)-approximate decision procedure.
Let \(P : [1, m] \to \R^d\) and \(Q : [1, n] \to \R^d\) be two polygonal chains in \(d\)-dimensional
Euclidean space as defined before, and let \(\alpha \in [\sqrt{n}, n]\).
Let \(\delta > 0\).
We begin by computing the grid~\(G\) along with \(O(n / \alpha)\) bad edges and points as defined in
Section~\ref{sec:preliminaries}.
We then explicitly compute and record a set of \(O(n^2 / \alpha^2)\) approximate reachability
intervals between dangerous vertices and bad edges.
To compute these intervals, we occasionally perform a linear time greedy search for a good
correspondence.
We describe this greedy search procedure in Section~\ref{sec:decision_procedure-greedy_mapping}
before giving the remaining details of the decision procedure in
Section~\ref{sec:decision_procedure-decision_details}.

\subsection{Greedy mapping subroutines}
\label{sec:decision_procedure-greedy_mapping}

We describe a pair of subroutines for greedily computing Fr\'echet correspondences along lengths of
\(P\) and \(Q\).
The first of these procedures \(\textsc{GreedyMappingP}(i, t)\) takes as its input an integer \(i
\in \Seq{1, \dots, m}\) such that \(p_i\) is a good vertex of \(P\) along with a real value \(t \in
[1, n]\) such that \(\dist(p_i, Q(t)) \leq \delta\).
Informally, the procedure does the following:
Suppose there exists a Fr\'echet correspondence \((\reparamP, \reparamQ)\) between \(P\) and \(Q\)
of cost at most \(\delta\) that maps \(p_i\) `close to' \(Q(t)\).
Procedure \(\textsc{GreedyMappingP}(i, t)\) essentially follows \(P\) and \(Q\) from box to box,
discovering groups of points that must be matched by \((\reparamP, \reparamQ)\).
When there is too much ambiguity in what must be matched to continue searching greedily, it outputs
a set of approximate reachability intervals, including one used by \((\reparamP, \reparamQ)\).
While we can infer which boxes pairs of matched points belong to, it may still be unclear exactly
which pairs appear in \((\reparamP, \reparamQ)\).
Also, the procedure may output intervals despite \((\reparamP, \reparamQ)\) not existing in the
first place!
Therefore, we can only guarantee the intervals can be reached using a correspondence of cost at most
\(O(\alpha)\cdot\delta\).
We define another procedure \(\textsc{GreedyMappingQ}(j, s)\) similarly, exchanging the roles of
\(P\) and \(Q\).
As they are rather technical, the precise definitions of these procedures are best expressed in the
following lemmas.
\begin{lemma}
  \label{lem:greedy_output}
  Let \(i \in \Seq{1, \dots, m}\) and \(t \in [1, n]\) such that \(p_i\) is good and \(\dist(p_i,
  Q(t)) \leq \delta\).
  Procedure \(\textsc{GreedyMappingP}(i, t)\) outputs zero or more approximate reachability
  intervals between a bad edge of \(P\) or \(Q\) and a dangerous vertex of \(Q\) or \(P\),
  respectively.
  For each pair \((s', t') \in [i, m] \times [t, n]\) in an approximate reachability interval
  computed by the procedure, there exists a Fr\'echet correspondence of cost
  \(O(\alpha)\cdot\delta\) between \(P[i, s']\) and \(Q[t, t']\).
  Procedure \(\textsc{GreedyMappingQ}(j, s)\) has the same properties with the roles of \(P\) and
  \(Q\) exchanged.
\end{lemma}
\begin{lemma}
  \label{lem:greedy_guarantee}
  Let \(i \in \Seq{1, \dots, m}\) and \(t \in [1, n]\) such that \(p_i\) is good and \(\dist(p_i,
  Q(t)) \leq \delta\).
  Suppose there exists a Fr\'echet correspondence \((\reparamP, \reparamQ)\) between \(P\)
  and \(Q\) of cost at most \(\delta\) that matches \(i\) with some \(t^* \geq t\) such that
  every point of \(Q[t, t^*]\) is at most distance \(3\delta\) from \(p_i\).
  Then, \((\reparamP, \reparamQ)\) passes through at least one approximate reachability interval
  output by procedure \(\textsc{GreedyMappingP}(i, t)\).
  Procedure \(\textsc{GreedyMappingQ}(j, s)\) has the same properties with the roles of \(P\) and
  \(Q\) exchanged.
\end{lemma}
We now provide details on the implementation of \(\textsc{GreedyMappingP}(i, t)\) along with
intuition for the steps it uses.
Procedure \(\textsc{GreedyMappingQ}(j, s)\) has an analogous description, with the roles of \(P\)
and \(Q\) exchanged.

To begin, observe \(p_i\) and \(Q(t)\) lie in the same box \(B\) of grid \(G\), because \(p_i\) is
good and \(\dist(p_i, Q(t)) \leq \delta\).
We first follow \(P\) and \(Q\) to see where they leave \(B\):
Let \(s^e = m\) if \(P\) never leaves \(B\) after \(p_i\).
Otherwise, let \(s^e\) be the minimum value in \((i, m]\) such that \(P(s^e)\) lies on the
boundary of \(B\) (the `\(e\)' stands for \emph{exit}).
Define \(t^e\) similarly for \(Q\).
See Figure~\ref{fig:greedy_definitions}.

If either curve ends before leaving \(B\), then the rest of the other curve needs to stay near \(B\)
if a correspondence like in Lemma~\ref{lem:greedy_guarantee} exists.
Therefore, if \(s^e = m\) (resp. \(t^e = n\)), we check if all points of \(Q[t, n]\) (resp. \(P[i,
m]\)) lie in or within distance \(\delta\) of \(B\).
If so, we output the trivial approximate reachability interval of \(\Set{(m, n)}\) and terminate the
procedure.
Otherwise, we output zero approximate reachability intervals.

From here on, we assume \(s^e \neq m\) and \(t^e \neq n\).
Let \(i^e \in \Seq{1, \dots, m}\) such that \(i^e - 1 \leq s^e \leq i^e\), and define \(j^e\)
similarly.
We begin by considering cases where a curve leaves box~\(B\) along a good edge.
Here, a correspondence as described in Lemma~\ref{lem:greedy_guarantee} must match a portion of the
other curve to the good edge.
Fortunately, we can guess approximately where that portion of the other curve begins and ends.
Afterward, the other endpoint of the good edge serves as a suitable parameter for a recursive call
to one of our greedy mapping procedures.

\begin{figure}[t]
\centering
\includegraphics[height=2in]{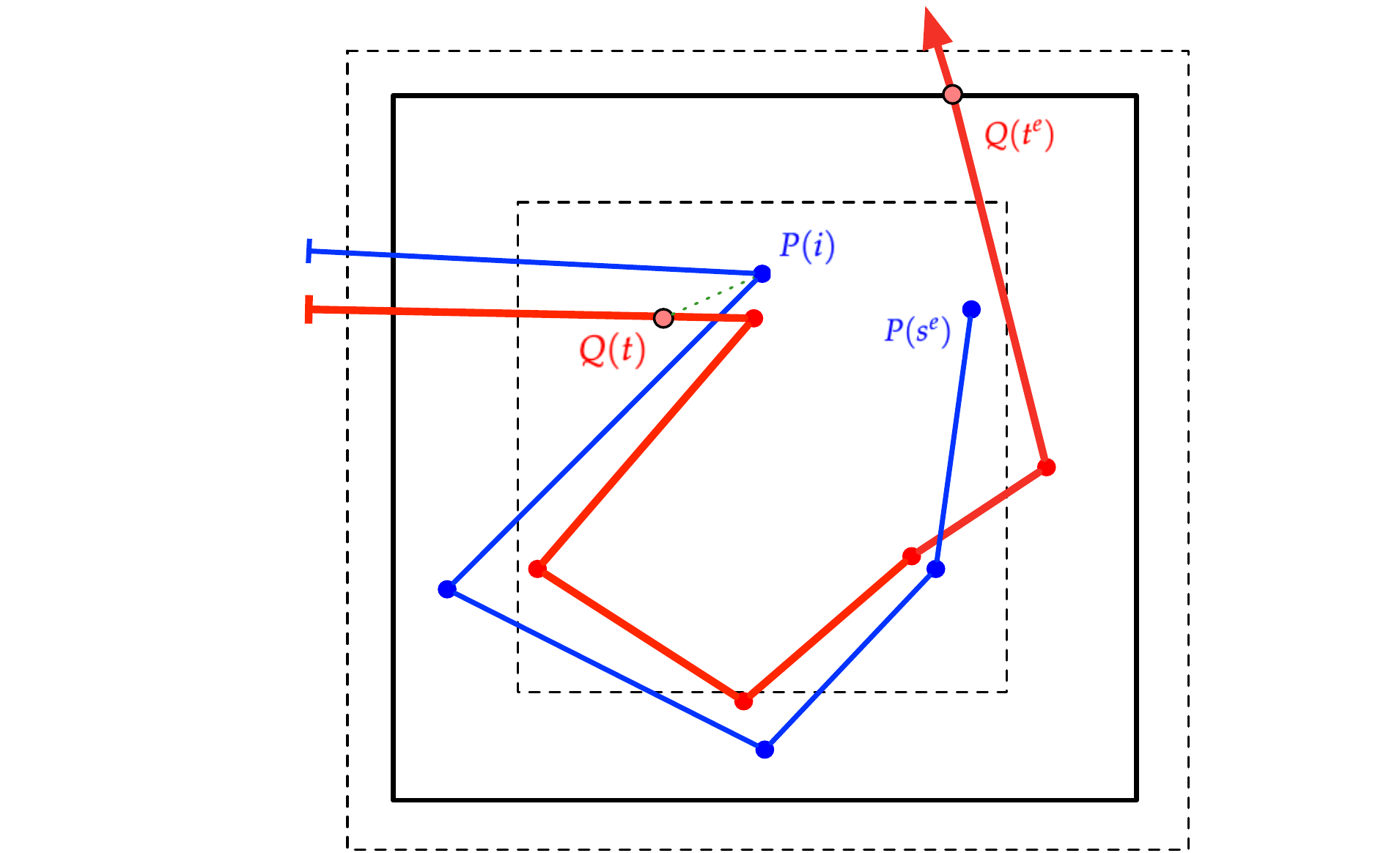}
\caption{Basic setup for \(\textsc{GreedyMappingP}(i, t)\)}
\label{fig:greedy_definitions}
\end{figure}

Specifically, suppose edge \(P[i^e-1, i^e]\) is good.
In this case, let \(t^f\) be the minimum value in \((t, n]\) such that \(\dist(p_{i^e},
Q(t^f)) \leq \delta\), and let \(t^c\) be the \emph{maximum} value in \([t, t^f)\) such
that \(\dist(p_{i^e-1}, Q(t^c)) \leq \delta\) (the `\(f\)' stands for \emph{far}, and the `\(c\)'
stands for \emph{close}).
See Figure~\ref{fig:greedy_good}.
We check if every point of \(Q[t, t^c]\) lies in or within distance \(\delta\) of \(B\)
and if \(\frechet(P[i^e-1, i^e], Q[t^c, t^f]) \leq \delta\).
If so, we run \(\textsc{GreedyMappingP}(i^e, t^f)\) and use its output.
Otherwise, we output zero approximate reachability intervals.

Now suppose the previous case does not hold but edge \(Q[j^e - 1, j^e]\) is good.
Here, we perform similar steps to those described in the previous case, exchanging the roles of
\(P\) and \(Q\).
Specifically, we let \(s^f\) be the minimum value in \((i, m]\) such that \(\dist(q_{j^e}, P(s^f))
\leq \delta\), and let \(s^c\) be the maximum value in \([i, s^f)\) such that \(\dist(q_{j^e-1},
Q(s^c)) \leq \delta\).
We check if every point of \(P[i, s^c]\) lies in or within distance \(\delta\) of \(B\) and if
\(\frechet(P[s^c, s^f], Q[j^e-1, j^e]) \leq \delta\).
If so, we run \(\textsc{GreedyMappingQ}(j^e, s^f)\) and use its output.
Otherwise, we output zero approximate reachability intervals.

\begin{figure}[t]
\centering
\hspace{-0.5in}
\includegraphics[height=2in,trim=1in 0 0 0]{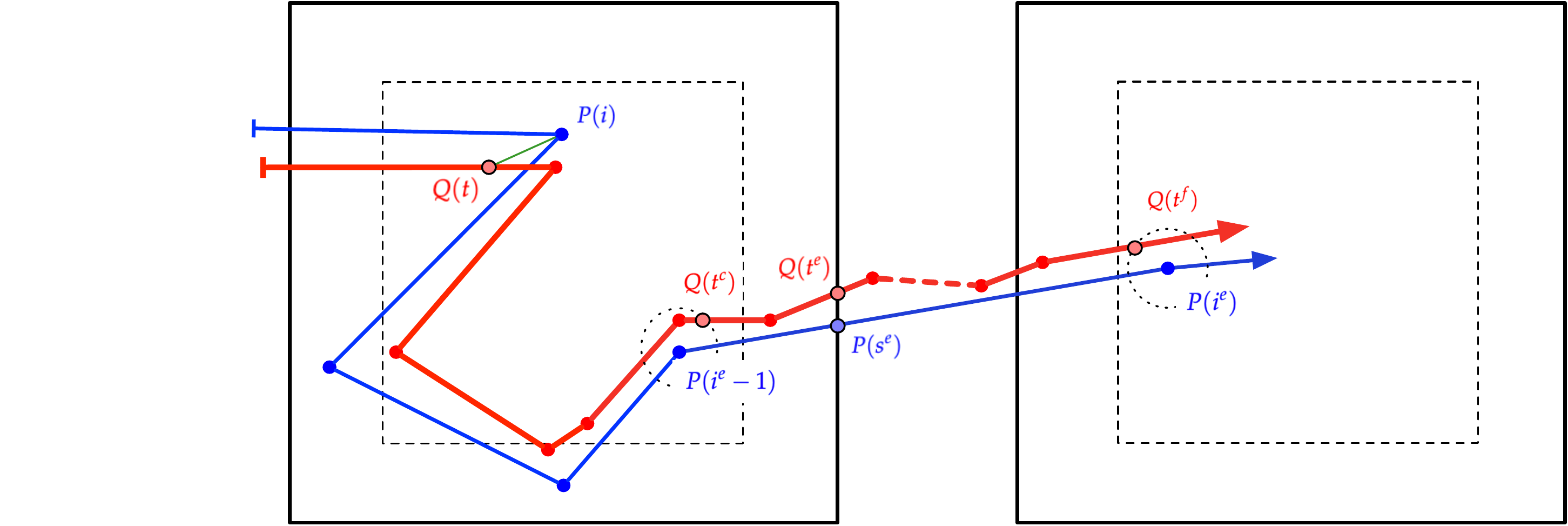}
\caption{\(\textsc{GreedyMappingP}(i, t)\): The case where \(P[i^e - 1, i^e]\) is good}
\label{fig:greedy_good}
\end{figure}

From here on, we assume neither curve leaves box \(B\) through a good edge.
Suppose there is a correspondence \((\reparamP, \reparamQ)\) as described in
Lemma~\ref{lem:greedy_guarantee}.
Further suppose the reparameterized walks along \(P\) and \(Q\) leave box \(B\) along \(P\) before
\(Q\).
In this case, we can show that \(P(s^e)\) is matched with a point on a bad edge of \(Q\).
Accordingly, we iterate over the bad edges of \(Q\) that appear before \(Q\) leaves box \(B\),
computing sufficiently large approximate reachability intervals along the top and right sides of
free space diagram cells for \(P[i^e - 1, i^e]\) and those bad edges of \(Q\).
Both of the edges for each of these cells are bad, so the intervals we compute are between bad edges
and dangerous vertices.

Specifically, let \(t \leq t_1 < t_2 < \dots < t_{\ell} \leq t^e\) be the list of first positions
along their respective edges of \(Q\) such that \(\dist(P(s^e), Q(t_{k})) \leq \delta\) for each \(k
\in \Seq{1, \dots, \ell}\).
See Figure~\ref{fig:greedy_bad}, left.
Observe that each edge containing a point \(t_k\) must be bad, because \(Q\) does not leave \(B\)
along a good edge, and no good edge with two endpoints in \(B\) lies within distance \(\delta\) of
\(P(s^e)\).
For each \(k \in \Seq{1, \dots, \ell}\), we do the following:
Let \(j_k \in \Seq{1, \dots, n}\) such that \(j_k - 1 \leq t_{k} \leq j_k\).
Let \(t^a_k\) be the minimum value in \([t_k, j_k]\) such that \(\dist(p_{i^e},
Q(t^a_k)) \leq \delta\) and let \(t^b_k\) be the maximum value in \([t_k, j_k]\) such that
\(\dist(p_{i^e}, Q(t^b_k)) \leq \delta\).
If \(t^a_k\) and \(t^b_k\) are well-defined, then we designate the interval \(\Set{i^e} \times
[t^a_k, t^b_k]\) as approximately reachable.
(If we have already designated a subset of \(\Set{i^e} \times [j_k - 1, j_k]\) as approximately
reachable earlier in the decision procedure, then we extend the approximately reachable area by
taking the union with the old interval.
Every interval of \(\Set{i^e} \times [j_k - 1, j_k]\) we compute will end at \((i^e, t^b_k)\), so
the union is also an interval.)
Similarly, let \(s^a_k\) be the minimum value in \([s^e, i^e]\) such that \(\dist(P(s^a_k),
q_{j_k}) \leq \delta\) and let \(s^b_k\) be the maximum value in \([s^e, i^e]\) such that
\(\dist(P(s^b_k), q_{j_k}) \leq \delta\).
If \(s^a_k\) and \(s^b_k\) are well-defined, then we designate the interval \([s^a_k, s^b_k] \times
\Set{j_k}\) as approximately reachable.
See Figure~\ref{fig:greedy_bad}, right.

It is also possible that a good correspondence has the walk along \(Q\) leave box \(B\) first.
So, \emph{in addition} to the above set of approximate reachability intervals, we also create some
based on points of \(P\) between \(p_i\) and \(P(s^e)\) that pass close to \(Q(t^e)\).
Let \(s \leq s_1 < s_2 < \dots < s_{\ell} \leq s^e\) be the exhaustive list of first
positions along their respective edges of \(P\) such that \(\dist(P(s_k), Q(t^e)) \leq \delta\)
for each \(k \in \Seq{1, \dots, \ell}\).
For each \(k \in \Seq{1, \dots, \ell}\), we do the following:
Let \(i_k \in \Seq{1, \dots, n}\) such that \(i_k - 1 \leq s_{k} \leq i_k\).
Let \(s^a_k\) be the minimum value in \([s_k, i_k]\) such that \(\dist(P(s^a_k), q_{j^e}) \leq
\delta\), and let \(s^b_k\) be the maximum value in \([s_k, i_k]\) such that \(\dist(P(s^b_k),
q_{j^e}) \leq \delta\).
If \(s^a_k\) and \(s^b_k\) are well-defined, then we designate the interval \([s^a_k, s^b_k] \times
\Set{j^e}\) as approximately reachable.
Similarly, let \(t^a_k\) be the minimum value in \([t^e, j^e]\) such that \(\dist(p_{j_k},
Q(t^a_k)) \leq \delta\),
and let \(t^b_k\) be the maximum value in \([t^e, j^e]\) such that
\(\dist(p_{j_k}, Q(t^b_k)) \leq \delta\).
If \(t^a_k\) and \(t^b_k\) are well-defined, then we designate the interval \(\Set{i_k} \times
[t^a_k, t^b_k]\) as approximately reachable.

We have concluded our description of \(\textsc{GreedyMappingP}(i, t)\) and are ready to prove
Lemmas~\ref{lem:greedy_output} and~\ref{lem:greedy_guarantee}.

\begin{figure}[t]
  \centering
  \hspace{-1in}
  \includegraphics[height=2.5in,trim=1in 0 0 0]{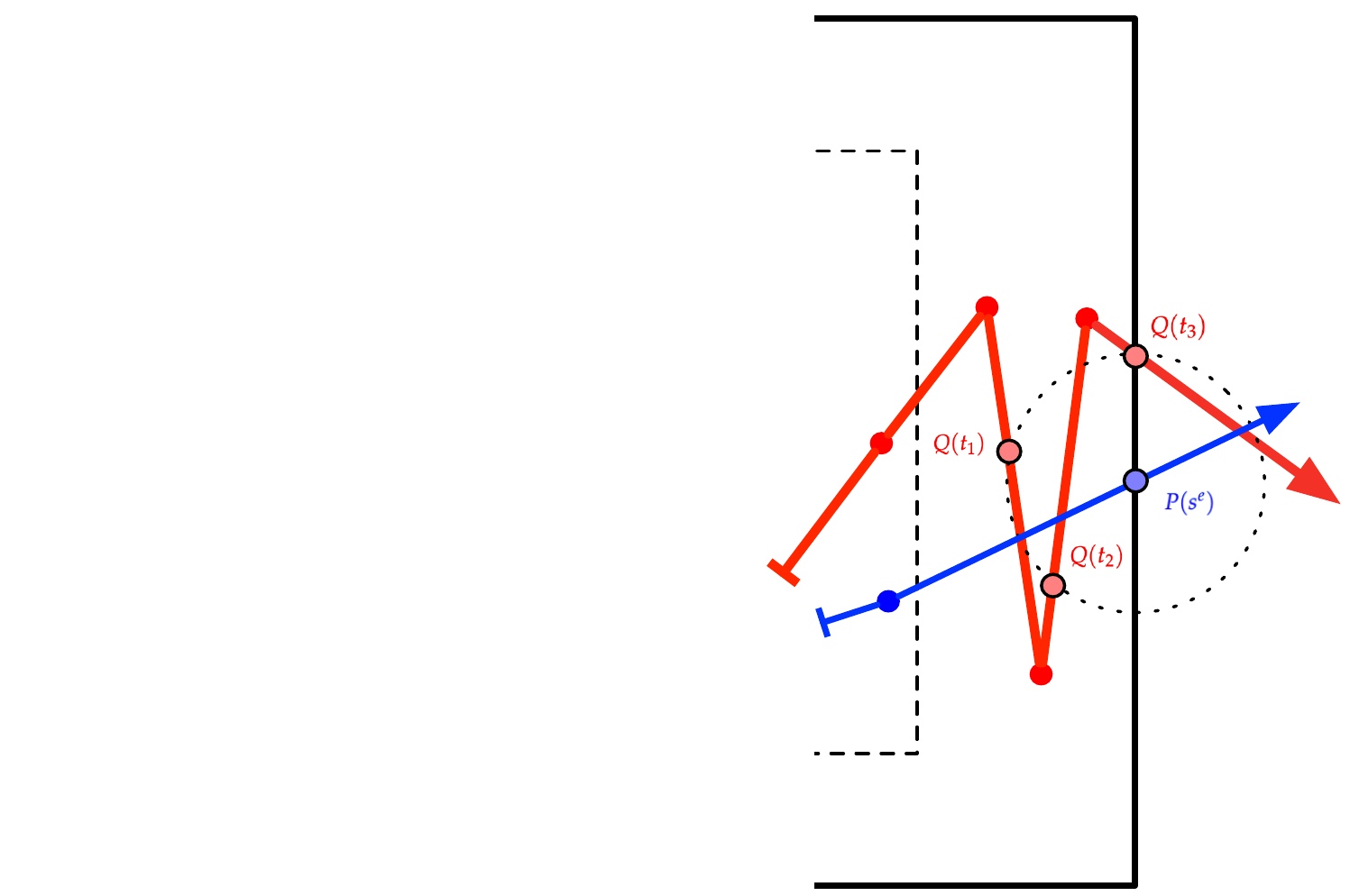}\qquad\qquad\qquad
  \includegraphics[height=2in]{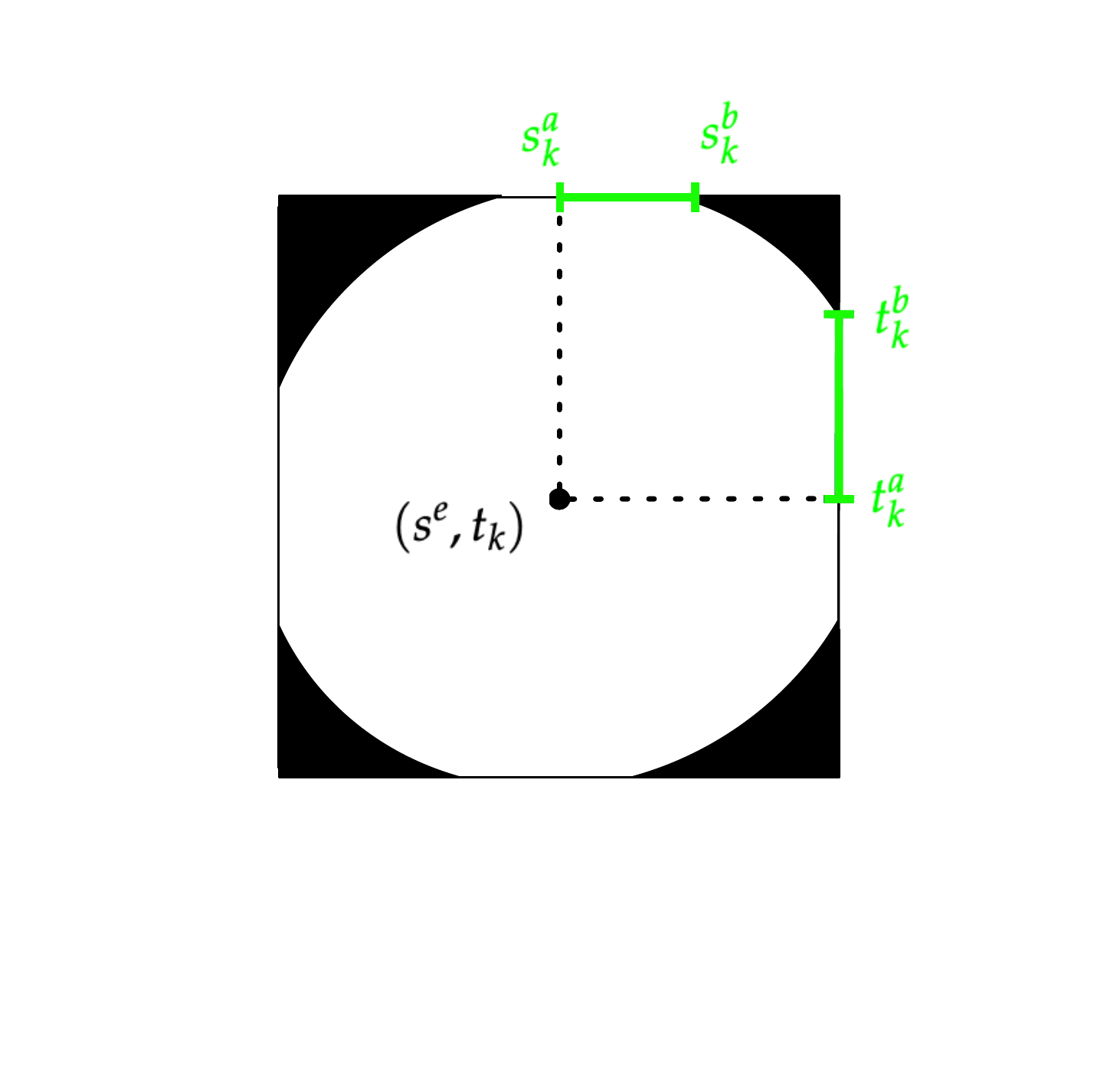}
  \caption{Left: Definition of \(\Seq{t_1, \dots, t_{\ell}}\).
  Right: Designating approximate reachability intervals near a pair \((s^e, t_k)\).
  The clipped ellipse coincides with the free points inside cell~\(C_{i^e, j_k}\).}
  \label{fig:greedy_bad}
\end{figure}

\begin{proof}[of Lemma~\ref{lem:greedy_output}]
  We use the same notation as given in the description of\linebreak
  \(\textsc{GreedyMappingP}(i, t)\).
  We first argue that we only output reachability intervals between bad edges and dangerous
  vertices.
  If we only output the trivial interval \(\Set{(m, n)}\) then the statement is trivially true.
  Otherwise, suppose we create an interval while working with \(P(s^e)\) and some nearby point
  \(Q(t_k)\).  We are not performing a recursive call to \(\textsc{GreedyMappingP}\) in this case,
  so \(P[i^e - 1, i^e]\) is bad, and \(p_{i^e}\) is dangerous.
  Similarly, we are not performing a recursive call to \(\textsc{GreedyMappingQ}\), so \(Q[j^k - 1,
  j^k]\) is not a good edge with endpoint \(q_{j^k}\) outside of box \(B\).
  Point \(Q(t_k)\) is within distance \(\delta\) of the boundary of \(B\), so \(Q[j_k - 1, j_k]\)
  cannot be a good edge with both endpoints in \(B\), either.
  We conclude \(Q[j_k - 1, j_k]\) is bad as well, and \(q_{j_k}\) is dangerous.
  A similar argument holds if we create an interval while working with \(Q(t^e)\) and some nearby
  point of \(P\).

  We now argue that for any pair of points \((s', t')\) on one of the approximate reachability
  intervals output by the procedure, there exists a correspondence of cost \(O(\alpha)\cdot \delta\)
  between \(P[i, s']\) and \(Q[t, t']\).
  First, suppose \(\textsc{GreedyMappingP}(i, t)\) creates one or more approximate reachability
  intervals without performing a recursive call.
  Suppose \(s^e = m\) or \(t^e = n\), implying \((s', t') = (m, n)\).
  All points of \(P[i, m]\) and \(Q[t, n]\) lie in or within distance \(\delta\) of \(B\), so they
  are all distance at most \(\sqrt{d} (\alpha+1) \cdot \delta\) from each other and any Fr\'echet
  correspondence between \(P[i, s']\) and \(Q[t, t']\) has cost \(O(\alpha) \cdot \delta\).

  Now suppose otherwise, but \((s', t')\) lies on an interval created while working with \(P(s^e)\)
  and some nearby point \(Q(t_k)\).
  All points of \(P[i, s^e]\) and \(Q[t, t_k]\) lie in \(B\), so they are all distance at most
  \(\sqrt{d} \alpha \cdot \delta\) from each other and any Fr\'echet correspondence between \(P[i,
  s^e]\) and \(Q[t, t_k]\) has cost \(O(\alpha) \cdot \delta\).
  The set of pairs \((x, y) \in P[i^e - 1, i^e] \times Q[j_k - 1, j_k]\) such that \(\dist(P(x),
  Q(y)) \leq \delta\) includes \((s^e, t_k)\) and \((s', t')\), and the set is
  convex~\cite{ag-cfdbt-95}, so we can extend our correspondence to include another between \(P[s^e,
  s']\) and \(Q[t_k, t']\) of cost at most \(\delta\).
  A similar argument covers the case where \((s', t')\) lies on an interval created while working
  with \(Q(t^e)\) and a nearby point of \(P\).

  Finally, suppose \(\textsc{GreedyMappingP}(i, t)\) recursively calls
  \(\textsc{GreedyMappingP}(i^e, t^f)\).
  Every point of \(P[i, i^e - 1]\) and \(Q[t, t^c]\) lies in or within distance \(\delta\) of \(B\),
  so every correspondence between \(P[i, i^e - 1]\) and \(Q[t, t^c]\) has cost at most
  \(O(\alpha)\cdot \delta\).
  Also, we have \(\frechet(P[i^e-1, i^e], Q[t^c, t^f]) \leq \delta\).
  We can combine these correspondences with the one inductively guaranteed by the call to
  \(\textsc{GreedyMappingP}(i^e, t^f)\) to get our desired correspondence between \(P[i, s']\) and
  \(Q[t, t']\).
  Again, a similar argument covers the case where we do a recursive call
  \(\textsc{GreedyMappingQ}(j^e, s^f)\).

  The proof for \(\textsc{GreedyMappingQ}(j, s)\) is the same, but with the roles of \(P\) and \(Q\)
  exchanged.
\end{proof}

\begin{proof}[of Lemma~\ref{lem:greedy_guarantee}]
  We use the same notation as given in the description of\linebreak
  \(\textsc{GreedyMappingP}(i, t)\).
  By assumption and the fact that \(p_i\) is good, every point of \(Q[t, t^*]\) lies within
  \(B\).
  Let \(r^{se} = \reparamP^{-1}(s^e)\) and \(r^{te} = \reparamQ^{-1}(t^e)\).

  Suppose \(\textsc{GreedyMappingP}(i, t)\) does not do a recursive call.
  If we output the trivial interval \(\Set{(m, n)}\), then the lemma is trivially true.
  Suppose we do not output the trivial interval and \(r^{se} \leq r^{te}\).
  Point \(Q(\reparamQ(r^{se}))\) lies on an edge \(Q[j_k - 1, j_k]\) with one of the points
  \(Q(t_k)\) where \(\dist(P(s^e), Q(t_k)) \leq \delta\).
  By definition of \(t_k\), we have \(t_k \leq \reparamQ(r^{se})\).
  Recall, the free space is convex within each individiaul cell of the free space
  diagram~\cite{ag-cfdbt-95}.
  Therefore, the set of \(s^e \leq s' \leq i^e\) such that \(\frechet(P[s^e, s'],
  Q[\reparamQ(r^{se}), j_k]) \leq \delta\) is precisely the approximate reachability interval
  \([s^a_k, s^b_k] \times \Set{j_k}\) we computed.
  Similarly, the set of \(\reparamQ(r^{se}) \leq t' \leq j_k\) such that \(\frechet(P[s^e, i^e],
  Q[\reparamQ(r^{se}), t']) \leq \delta\) is actually a \emph{suffix} of the approximate
  reachability interval \(\Set{i^e} \times [t^a_k, t^b_k]\) we computed.
  A similar argument applies if \(r^{te} < r^{se}\).

  Finally, suppose \(\textsc{GreedyMappingP}(i, t)\) recursively calls
  \(\textsc{GreedyMappingP}(i^e, t^f)\).
  Let \(t^{f*}\) be matched with \(i^e\) and \(t^{c*}\) be matched with \(i^e - 1\) by \((\reparamP,
  \reparamQ)\).
  Because \(p_{i^e - 1}\) and \(p_{i^e}\) are both good, \(\dist(p_{i^e - 1}, Q(t^{c*})) \leq
  \delta\), and \(\dist(p_{i^e}, Q(t^f)) \leq \delta\), points \(Q(t^{c*})\) and \(Q(t^f)\) lie
  within the same boxes as \(p_{i^e - 1}\) and \(p_{i^e}\), respectively.
  These boxes are distinct, so we may conclude \(t^{c*} \leq t^f\).
  Further, we chose \(t^c \geq t^{c*}\) and \(t^f \leq t^{f*}\), and we may infer \(Q(t^c)\) and
  \(Q(t^{f*})\) also lie in the same boxes as \(p_{i^e - 1}\) and \(p_{i^e}\), respectively.
  We conclude \(t^{c*} \leq t^c < t^f \leq t^{f*}\).

  Consider the following correspondence between \(P[i^e - 1, i^e]\) and \(Q[t^c, t^f]\):
  Let \(s^c \geq i^e - 1\) and \(s^f \leq i^e\) be matched to \(t^c\) and \(t^f\), respectively, by
  \((\reparamP, \reparamQ)\).
  We match every point of \(P[i^e - 1, s^c]\) to \(Q(t^c)\), match \(P[s^c, s^f]\) to \(Q[t^c,
  t^f]\) exactly as done by \((\reparamP, \reparamQ)\), and match every point of \(P[s^f, i^e]\) to
  \(Q(t^f)\).
  See Figure~\ref{fig:greedy_remap}.
  We have \(\dist(p_{i^e - 1}, Q(t^c)) \leq \delta\) and \(\dist(P(s^c), Q(t^c)) \leq \delta\), so
  the entire line segment \(P[i^e - 1, s^c]\) lies within distance \(\delta\) of \(Q(t^c)\).
  Similarly, the line segment \(P[s^f, i^e]\) lies within distance \(\delta\) of \(Q(t^f)\).
  Our correspondence has cost at most \(\delta\).

  \begin{figure}[t]
    \centering
    \includegraphics[height=2in,trim=2in 0 1in 0]{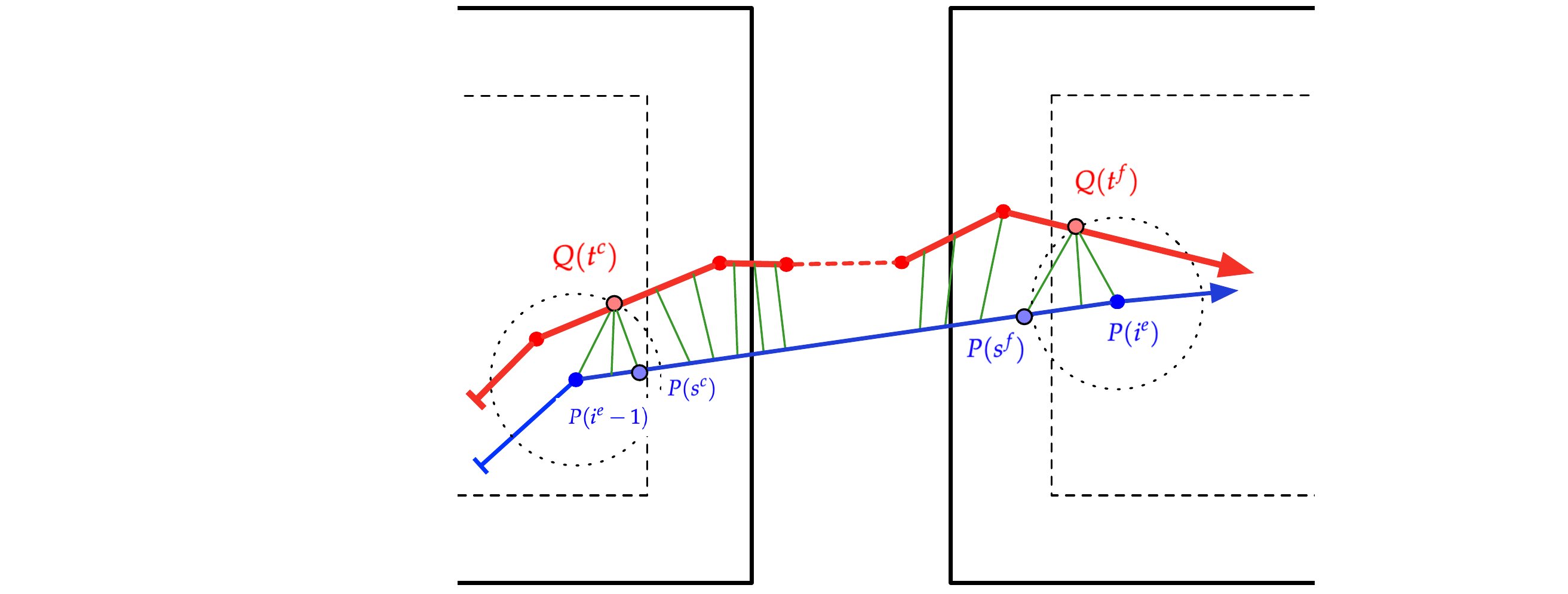}
    \caption{A correspondence of cost \(\delta\) between \(P[i^e - 1, i^e]\) and \(Q[t^c, t^f]\).
    A subset of matched points are represented by thin green line segments.}
    \label{fig:greedy_remap}
  \end{figure}

  Now, consider any point \(Q(t')\) with \(t^{c*} \leq t' \leq t^c\) and let \(s'\) be matched to
  \(t'\) by \((\reparamP, \reparamQ)\).
  We have \(\dist(P(s'), Q(t')) \leq \delta\).
  We argued that line segment \(P[i^e - 1, s^c]\) is within distance \(\delta\) of \(Q(t^c)\),
  implying \(\dist(P(s'), Q(t^c)) \leq \delta\).
  Finally, \(\dist(p_{i^e - 1}, Q(t^c)) \leq \delta\).
  By triangle inequality, \(\dist(p_{i^e - 1}, Q(t')) \leq 3\delta\), implying \(Q(t')\) lies in
  \(B\).
  As explained above, every point of \(Q[t, t^*]\) lies in \(B\).
  Also, every point of \(Q[t^*, t^{c*}]\) lies within distance \(\delta\) of a point in
  \(P[i, i^e - 1]\) and therefore lies in or within distance \(\delta\) of \(B\).
  And, we just showed every point of \(Q[t^{c*}, t^c]\) lies in \(B\).
  Our algorithm will succeed at all its distance checks and recursively call
  \(\textsc{GreedyMappingP}(i^e, t^f)\).
  Finally, a similar triangle inequality argument implies every point of \(Q[t^f, t^{*f}]\) is at
  most distance \(3\delta\) from \(p_{i^e}\).
  We are inductively guaranteed that \((\reparamP, \reparamQ)\) passes through an
  approximate reachability interval output during the recursive call.
  Similar arguments apply if \(\textsc{GreedyMappingP}(i, t)\) does a recursive call
  \(\textsc{GreedyMappingQ}(j^e, s^f)\).

  The proof for \(\textsc{GreedyMappingQ}(j, s)\) is the same as that given above, but with the
  roles of \(P\) and \(Q\) exchanged.
\end{proof}

\subsection{Remaining decision procedure details}
\label{sec:decision_procedure-decision_details}

We now fill in the remaining details of our approximate decision procedure.
Recall, we have computed a grid \(G\) with boxes of side length \(\alpha\cdot\delta\) such that
there are \(O(n / \alpha)\) bad vertices of \(P\) and \(Q\).
Also recall, \(p_1\), \(p_m\), \(q_1\), and \(q_n\) are designated as bad regardless of their
position in \(G\)'s boxes.
As described below, our decision procedure, iteratively in lexicographic order, checks each cell of
the free space diagram for which we may have computed an approximate reachablity interval on its
left or bottom side.
We then extend the known approximately reachable space from each non-empty interval in one of two
ways.
Depending on whether relevant edges are good or bad, we either perform a call to the appropriate
greedy mapping subroutine to seek out new intervals that are approximately reachable but potentially
far away in the free space diagram, or we directly compute approximate reachability intervals on the
right or top sides of the cell using the constant time method of Alt and Godau~\cite{ag-cfdbt-95}.

Specifically, we first check if \(\dist(p_1, q_1) \leq \delta\).
If not, our procedure reports failure.
Otherwise, let \(t^b\) and \(s^b\) be the maximum values in \([1, 2]\) such that \(\dist(p_1,
Q(t^b)) \leq \delta\) and \(\dist(P(s^b), q_1) \leq \delta\), respectively.
We designate intervals \(\Set{1} \times [1, t^b]\) and \([1, s^b] \times \Set{1}\) as
(approximately) reachable.
Now, for each \(i \in \Seq{2, \dots, m}\) such that \(p_{i - 1}\) is dangerous, for each \(j \in
\Seq{2, \dots, n}\) such that \(q_{j - 1}\) is dangerous, we do the following.

Suppose we have designated an interval \(\Set{i-1} \times [t^a, t^b]\) as approximately reachable
where \(j-1 \leq t^a \leq t^b \leq j\).
Suppose edge \(P[i - 1, i]\) is good.
Then, we run the procedure \(\textsc{GreedyMappingP}(i - 1, t^a)\).
If edge \(P[i - 1, i]\) is bad, we compute new approximate reachability intervals more directly as
follows.
First, let \(t^{a'}\) be the minimum value in \([t^a, j]\) such that \(\dist(p_i, Q(t^{a'})) \leq
\delta\), and let \(t^{b'}\) be the maximum value in \([t^a, j]\) such that \(\dist(p_i, Q(t^{b'}))
\leq \delta\).
We designate interval \(\Set{i} \times [t^{a'}, t^{b'}]\) as approximately reachable (again, we may
end up extending a previously computed approximately reachability interval on \(\Set{i} \times [j -
1, j]\)).
Similarly, let \(s^{a'}\) be the minimum value in \([i - 1, i]\) such that \(\dist(P(s^{a'}), q_j)
\leq \delta\), and let \(s^{b'}\) be the maximum value in \([i - 1, i]\) such that
\(\dist(P(s^{a'}), q_j) \leq \delta\).
We designate interval \([s^{a'}, s^{b'}] \times \Set{j}\) as approximately reachable.
We are done working with interval \(\Set{i-1} \times [t^a, t^b]\).

Now, suppose we have designated interval \([s^a, s^b] \times \Set{j - 1}\) as approximately
reachable where \(i - 1 \leq s^a \leq s^b \leq i\).
Suppose edge \(Q[j - 1, j]\) is good.
If so, we run the procedure \(\textsc{GreedyMappingQ}(j - 1, s^a)\).
If edge \(Q[j - 1, j]\) is bad, we compute new approximate reachability intervals more directly as
follows.
First, let \(t^{a'}\) be the minimum value in \([j - 1, j]\) such that \(\dist(p_i, Q(t^{a'})) \leq
\delta\), and let \(t^{b'}\) be the maximum value in \([j - 1, j]\) such that \(\dist(p_i,
Q(t^{b'})) \leq \delta\).
We designate interval \(\Set{i} \times [t^{a'}, t^{b'}]\) as approximately reachable.
Similarly, let \(s^{a'}\) be the minimum value in \([s^a, i]\) such that \(\dist(P(s^{a'}), q_j)
\leq \delta\), and let \(s^{b'}\) be the maximum value in \([s^b, i]\) such that \(\dist(P(s^{a'}),
q_j) \leq \delta\).
We designate interval \([s^{a'}, s^{b'}] \times \Set{j}\) as approximately reachable.
We are done working with interval \([s^a, s^b] \times \Set{j - 1}\).

Once we have completed the iterations, we do one final step.
We check if \((m, n)\) lies on an approximate reachability interval.
If so, we report there is a Fr\'echet correspondence between \(P\) and \(Q\) of cost
\(O(\alpha)\cdot \delta\).
Otherwise, we report failure.

The following lemmas establish the correctness and running time for our decision procedure.

\begin{lemma}
  \label{lem:decision_intervals}
  The approximate decision procedure creates approximate reachability intervals only between bad
  edges of \(P\) or \(Q\) and dangerous vertices of \(Q\) or \(P\), respectively.
\end{lemma}
\begin{proof}
  Vertices \(p_1\) and \(q_1\) are bad, so the intervals we compute before beginning the for loops
  are between bad edges and dangerous vertices.
  Now, consider working with some approximate reachability interval \(\Set{i - 1} \times [t^a,
  t^b]\) with \(j - 1 \leq t^a \leq t^b \leq j\).
  Inductively, we may assume \(Q[j - 1, j]\) is bad, implying \(q_j\) is dangerous.
  If \(P[i - 1, i]\) is good, then Lemma~\ref{lem:greedy_output} guarantees we only create
  approximate reachability intervals between bad edges and dangerous vertices.
  Otherwise, \(p_i\) is dangerous, and both approximate reachability intervals we directly create
  are for bad edge/dangerous vertex pairs.
  A similar argument applies when working with some interval \([s^a, s^b] \times \Set{j - 1}\).
\end{proof}

\begin{lemma}
  \label{lem:decision_output}
  The approximate decision procedure is correct if it reports \(\frechet(P, Q) \leq O(\alpha) \cdot
  \delta\).
\end{lemma}
\begin{proof}
  Let \((s', t')\) be any member of an approximate reachability interval created by the procedure.
  We will show there exists a Fr\'echet correspondence between \(P[1, s']\) and \(Q[1, t']\) of cost
  \(O(\alpha) \cdot \delta\).
  Setting \((s', t') = (m, n)\) then proves the lemma.
  First, if \((s', t')\) lies on either interval created before the for loops begin, there is a
  trivial correspondence between \(P[1, s']\) and \(Q[1, t']\) of cost at most \(\delta\) that only
  uses one point of either \(P\) or \(Q\).
  Now, consider working with some approximate reachability interval \(\Set{i - 1} \times [t^a,
  t^b]\) with \(j - 1 \leq t^a \leq t^b \leq j\).
  Inductively, we may assume there is a correspondence of cost \(O(\alpha)\cdot \delta\) between
  \(P[1, i - 1]\) and \(Q[1, t^a]\).

  Suppose \(P[i - 1, i]\) is good, and we call \(\textsc{GreedyMappingP}(i - 1, t^a)\).
  By Lemma~\ref{lem:greedy_output}, we can extend our inductively guaranteed correspondence to one
  of cost \(O(\alpha)\cdot \delta\) ending at any point \((s', t')\) in any approximate reachability
  interval output by \(\textsc{GreedyMappingP}(i - 1, t^a)\).
  Now, suppose instead that \(P[i - 1, i]\) is bad.
  As in the proof of Lemma~\ref{lem:greedy_output} or the original exact algorithm of Alt and
  Godau~\cite{ag-cfdbt-95}, there is a Fr\'echet correspondence of cost at most \(\delta\) between
  \(P[i - 1, s']\) and \(Q[t^a, t']\) for any \((s', t')\) on the approximate reachability intervals
  we directly compute.
  Again, we can extend the inductively guaranteed correspondence to end at any such \((s', t')\).
  A similar argument applies when working with some interval \([s^a, s^b], \times \Set{j - 1}\).
\end{proof}

\begin{lemma}
  \label{lem:decision_guarantee}
  Suppose there exists a Fr\'echet correspondence \((\reparamP, \reparamQ)\) between \(P\) and \(Q\)
  of cost at most \(\delta\).
  The approximate decision procedure will report \(\frechet(P, Q) \leq O(\alpha) \cdot \delta\).
\end{lemma}
\begin{proof}
  Suppose \((\reparamP, \reparamQ)\) matches a pair \((i - 1, t^*)\) on some approximate
  reachability interval \(\Set{i - 1} \times [t^a, t^b]\).
  Suppose \(P[i - 1, i]\) is good.
  Every point of \(Q[t^a, t^*]\) lies within distance \(\delta\) of \(p_{i - 1}\).
  Lemma~\ref{lem:greedy_guarantee} guarantees \(\textsc{GreedyMappingP}(i - 1, t^a)\) will output at
  least one approximate reachability interval which includes a matched pair of \((\reparamP,
  \reparamQ)\).
  We can easily verify that the interval must involve a later vertex of \(P\) than \(p_{i - 1}\).

  Now, suppose instead that \(P[i - 1, i]\) is bad.
  The set of \(i - 1 \leq s' \leq i\) such that \(\frechet(P[i - 1, s'], Q[t^*, j]) \leq \delta\) is
  precisely the approximate reachability interval \([s^{a'}, s^{b'}] \times \Set{j}\) we computed.
  Similarly, the set of \(t^* \leq t' \leq j\) such that \(\frechet(P[i - 1, i], Q[t^*, t']) \leq
  \delta\) is actually a suffix of the approximate reachability interval \(\Set{i} \times
  [t^{a'}, t^{b'}]\) we computed.

  Either way, we have \((\reparamP, \reparamQ)\) using an interval for a later vertex of \(P\) or
  \(Q\).
  If the interval contains \((m, n)\), the decision procedure will report there exists a cheap
  correspondence.
  Otherwise, we may assume it will report one inductively.
  Similar arguments apply if \((\reparamP, \reparamQ)\) includes a point on some approximate
  reachability interval \([s^a, s^b] \times \Set{j - 1}\).

  Finally, we observe that \((\reparamP, \reparamQ)\) does include a point on at least one
  approximate reachability interval, because our procedure begins by computing two intervals that
  include \((1, 1)\).
\end{proof}

\begin{lemma}
  \label{lem:greedy_time}
  Procedures \(\textsc{GreedyMappingP}(i, t)\) and \(\textsc{GreedyMappingQ}(j, s)\) can be
  implemented to run in \(O(n)\) time.
\end{lemma}
\begin{proof}
  We use the notation given in the description of \(\textsc{GreedyMappingP}\).
  Let \(m' = m - i + 1\), and let \(n'\) be the number of vertices remaining in \(Q\) after
  \(Q(t)\).
  If \(s^e = m\) or \(t^e = n\), then we spend \(O(m' + n')\) time checking if a suffix of \(P\) and
  \(Q\) lies in or near box \(B\).
  From here, assume neither \(s^e = m\) nor \(t^e = n\).

  Suppose edge \(P[i^e - 1, i^e]\) is good.
  Let \(m'' = i^e - i \geq 1\), and let \(n''\) be the number of vertices in \(Q[t, t^f]\).
  We need to scan \(P\) and \(Q\) to find \(i^e\), \(t^c\), and \(t^f\).
  We also need to check if every point of \(Q[t, t^c]\) lies in or close to \(B\).
  Doing these steps takes \(O(m'' + n'')\) time.
  We need to check if \(\frechet(P[i^e - 1, i^e], Q[t^c, t^f]) \leq \delta\).
  The portion of \(P\) in this check consists of a single line segment, so it can be done in
  \(O(n'')\) time.
  Finally, we do a recursive call to \(\textsc{GreedyMappingP}(i^e, t^f)\) that inductively takes
  \(O(n' + m' - n'' - m'')\) time.
  In total, we spend \(O(n' + m')\) time.
  A similar argument holds if \(P[i^e - 1, i^e]\) is bad but \(Q[j^e - 1, j^e]\) is good.

  Finally, suppose both edges are bad.
  We spend \(O(n' + m')\) time total searching for \(s^e\) and \(t^e\), finding points from the
  other curve that lie close to \(s^e\) and \(t^e\), and computing approximate reachability
  intervals for each of these pairs of points.
\end{proof}

\begin{lemma}
  \label{lem:decision_time}
  The approximate decision procedure can be implemented to run in \(O(n^3 / \alpha^2)\) time.
\end{lemma}
\begin{proof}
  Finding the grid \(G\) with the set of \(O(n / \alpha)\) bad vertices takes \(O(n)\)
  time~\cite[Lemma 1]{cr-iaadf-18}.
  There are at most twice as many bad edges as bad vertices, and at most twice as many dangerous
  vertices as bad edges, so there are \(O(n / \alpha)\) dangerous vertices.
  Therefore, the decision procedure iterates over \(O(n^2 / \alpha^2)\) values of \(i\) and \(j\).
  For each pair, we do at most two \(O(n)\) time calls to \(\textsc{GreedyMappingP}\) or
  \(\textsc{GreedyMappingQ}\), or we compute up to four approximate reachability intervals directly
  in constant time each.
\end{proof}

Our decision procedure is easily extended to actually output a correspondence of cost
\(O(\alpha)\cdot \delta\) instead of merely determining if one exists by concatenating the smaller
correspondences we discover directly during the iterations or during runs of
\(\textsc{GreedyMappingP}\) and \(\textsc{GreedyMappingQ}\) as we compute approximate reachability
intervals.
We are now able to state the main result of this section.
\begin{lemma}
  \label{lem:decision}
  Let \(P\) and \(Q\) be two polygonal chains in \(\R^d\) of at most \(n\) vertices each, let
  \(\alpha \in [\sqrt{n}, n]\), and let \(\delta \geq 0\) be a parameter.
  We can compute a Fr\'echet correspondence between \(P\) and \(Q\) of cost at most \(O(\alpha)
  \cdot \delta\) or verify that \(\frechet(P, Q) > \delta\) in \(O(n^3 / \alpha^2)\) time.
\end{lemma}

\section{The Approximation Algorithm}
\label{sec:approximation_algorithm}

We now describe how to turn our approximate decision procedure into an approximation algorithm whose
approximation ratio is arbitrarily close to that of the decision procedure.
We emphasize that our techniques use the decision procedure as a black box subroutine, so any
improvement to the running time of our approximate decision procedure will imply the same
improvement to our approximation algorithm.
In short, we use our approximate decision procedure to binary search over a set of \(O(n)\)
distances approximating the distances between vertices of \(P\) and \(Q\).
If the Fr\'echet distance lies in a large enough gap between a pair of these approximate distances,
then we can simplify both polygonal chains so that their edge lengths become large compared to their
Fr\'echet distance.
We then run an exact Fr\'echet distance algorithm of Gudmundsson \etal~\cite{gmmw-ffdcl-19} designed
for this case.

Let \(P : [1, m] \to \R^d\) and \(Q : [1, n] \to \R^d\) be two polygonal chains in \(d\)-dimensional
Euclidean space, and suppose we have an approximate decision procedure for the Fr\'echet distance
between two polygonal chains with approximation ratio \(\alpha\).
We assume \(\alpha\) is at most a polynomial function of \(n\) (although it may be constant).
Let \(T(n, \alpha)\) denote the worst-case running time of the procedure on two polygonal chains of
at most \(n\) vertices each.
We assume \(T(n, \alpha) = \Omega(n)\).
Finally, consider any \(0 < \eps \leq 1\).
We describe how to compute an \(O((1 + \eps)\alpha)\)-approximation of \(\frechet(P, Q)\) in
\(O(T(n, \alpha) \log (n / \eps))\) time.

We begin by performing a binary search over a set~\(Z\) of \(O(n)\) values close to all of the
distances between pairs of vertices in \(P\) and \(Q\).
Let \(V\) denote the set of vertex points in \(P\) and \(Q\).
Our set \(Z\) is such that for any pair of distinct points \(o_1, o_2 \in V\), there exist \(x, x'
\in Z\) such that \(x \leq \dist(o_1, o_2) \leq x' \leq 2x\).
Such a set can be computed in~\(O(n \log n)\) time~\cite[Lemma 3.9]{dhw-afdrc-12}.
To perform the binary search, we simply search ``down'' if the approximate decision procedure
finds an \(\alpha\)-approximate correspondence, and we search ``up'' if it does not.
Let \(a\) and \(b\) be the largest value of \(Z\) for which the procedure fails and
the smallest value for which it succeeds, respectively.
If \(a\) does not exist, then we return the correspondence of cost \(\alpha \cdot b\)
found for \(b\).
We are guaranteed \(b\) exists, because the maximum distance between \(P\) and \(Q\) is achieved at
a pair of vertices.
From here on, we assume \(a\) exists.

We check if the approximate decision procedure finds a correspondence when given parameter \(\delta
:= 12 a / \eps\).
If so, let \(Z^a\) denote the sequence of distances \(\Seq{(1+\eps)^0\cdot a, (1+\eps)^1 \cdot a,
\dots, (1+\eps)^{\Ceil{12 / \eps}} \cdot a}\).
We binary search over \(Z^a\) and return the cheapest correspondence found.

Suppose no correspondence is found for \(12 a / \eps\).
We check if the approximate decision procedure finds a correspondence when given parameter
\(\delta := b / (2(1+\eps/2)(1+\sqrt{d})\alpha)\).
If \emph{not}, let \(Z^b\) denote the sequence of distances \(\Seq{b / (1+\eps)^0, b /
(1+\eps)^1, \dots, b / (1+\eps)^{\Ceil{2(1+\eps/2)(1+\sqrt{d})\alpha}}}\).
We binary search over \(Z^b\) and return the cheapest correspondence found.

Finally, suppose no correspondence is found for \(12a / \eps\) but one is found for \(b /
(2(1+\eps/2)(1+\sqrt{d})\alpha)\).
We perform a \(3a\)-simplification of \(P\) and \(Q\), yielding the polygonal chains
\(\hat{P}\) and \(\hat{Q}\) with at most \(n\) vertices each.
Gudmundsson \etal~\cite{gmmw-ffdcl-19} describe an \(O(n \log n)\) time algorithm that computes the
Fr\'echet distance of two polygonal chains \emph{exactly} if all of their edges have length at least
\((1 + \sqrt{d})\) times their Fr\'echet distance.
Their algorithm will succeed in finding an optimal Fr\'echet correspondence between \(\hat{P}\) and
\(\hat{Q}\).
This correspondence can be modified to create one for \(P\) and \(Q\) of cost at most \((1 + \eps)
\alpha \cdot \frechet(P, Q)\) (see Driemel~\etal~\cite[Lemmas 2.3 and 3.5]{dhw-afdrc-12}).

\begin{lemma}
  \label{lem:approximation_guarantee}
  The approximation algorithm finds a correspondence between \(P\) and \(Q\) of cost at most \((1 +
  \eps)\alpha \cdot \frechet(P, Q)\).
\end{lemma}
\begin{proof}
  Suppose value \(a\) as defined in the procedure does not exist.
  We find a correspondence of cost at most \(\alpha \cdot b \leq \alpha \cdot \dist(p_1, q_1) \leq
  \alpha \cdot \frechet(P, Q)\).
  We assume from here on that \(a\) exists.

  Suppose a binary search over \(Z^a\) or \(Z^b\) is performed.
  There exist values \(a'\) and \(b' = (1 + \eps)a'\) such that the approximate decision procedure
  fails with \(a'\) but succeeds at finding a correspondence of cost at most \(\alpha \cdot b'\).
  We have \(\alpha \cdot b' = (1+\eps)\alpha \cdot a' < (1 + \eps)\alpha \cdot \frechet(P, Q)\).

  Finally, suppose we perform binary searches over neither \(Z^a\) nor \(Z^b\).
  In this case, we observe \(12 a / \eps < \frechet(P, Q) \leq b / (2(1+\eps/2)(1+\sqrt{d}))\).
  Every distance between a pair of vertices in \(P\) or \(Q\) is either at most \(2a < (\eps / 6)
  \frechet(P, Q)\) or at least \(b / 2 \geq (1 + \sqrt{d})(1+\eps/2)\frechet(P, Q)\).
  We observe \(\frechet(\hat{P}, \hat{Q}) \leq \frechet(P, Q) + 6a < (1+\eps/2)\frechet(P,
  Q)\)~\cite[Lemma 2.3]{dhw-afdrc-12}.
  Polygonal chains \(\hat{P}\) and \(\hat{Q}\) have no edges of length at most \(2a\), implying
  all edges have length at least \((1+\sqrt{d})(1+\eps/2)\frechet(P, Q) >
  (1+\sqrt{d})\frechet(\hat{P}, \hat{Q})\).
  The conditions for the algorithm of Gudmundsson \etal~\cite{gmmw-ffdcl-19} are met, and as
  explained earlier, their algorithm will lead to the desired correspondence between \(P\) and
  \(Q\).
\end{proof}

\begin{lemma}
  \label{lem:approximation_time}
  The approximation algorithm can be implemented to run in \(O(T(n, \alpha) \log (n / \eps))\) time.
\end{lemma}
\begin{proof}
  We spend \(O(n \log n)\) time computing \(Z\).
  We do \(O(\log n)\) calls to the approximate decision procedure binary searching over \(Z\).
  Sequences \(Z^a\) and \(Z^b\) contain \(O(\log_{1+\eps} (1 / \eps)) = O((1 / \eps) \log (1 /
  \eps))\) and \(O(\log_{1+\eps} \alpha) = O((1 / \eps) \log n)\) values, respectively.
  Therefore, binary searching over \(Z^a\) or \(Z^b\) requires \(O(\log ((1 / \eps) \log (n /
  \eps))) = O(\log (n / \eps))\) calls to the approximate decision procedure.
  The case where we have to simplify the polygonal chains and run the algorithm of Gudmundsson
  \etal~\cite{gmmw-ffdcl-19} requires only \(O(n \log n)\) additional time.
  The lemma follows.
\end{proof}

We may now state the main result of this section.

\begin{theorem}
  \label{thm:approx_from_decision}
  Suppose we have an \(\alpha\)-approximate decision procedure for Fr\'echet distance that runs in
  time \(T(n, \alpha)\) on two polygonal chains in \(\R^d\) of at most \(n\) vertices each.
  Let \(0 < \eps \leq 1\).
  Given two such chains \(P\) and \(Q\), we can find a Fr\'echet correspondence between \(P\) and
  \(Q\) of cost at most \((1+\eps)\alpha\cdot\frechet(P, Q)\) in \(O(T(n, \alpha) \log (n / \eps))\)
  time.
\end{theorem}

Combining Theorem~\ref{thm:approx_from_decision} with Lemma~\ref{lem:decision} while
setting \(\eps := 1\) gives us our main result.

\begin{corollary}
  Let \(P\) and \(Q\) be two polygonal chains in \(\R^d\) of at most \(n\) vertices each, and let
  \(\alpha \in [\sqrt{n}, n]\).
  We can compute a Fr\'echet correspondence between \(P\) and \(Q\) of cost at most \(O(\alpha)
  \cdot \frechet(P, Q)\) in \(O((n^3 / \alpha^2) \log n)\) time.
\end{corollary}

\section{Conclusion}
\label{sec:conclusion}

We described the first strongly subquadratic time approximation algorithm for the continuous
Fr\'echet distance that has a subexponential approximation guarantee.
Specifically, it computes an \(O(\alpha)\)-approximate Fr\'echet correspondence in \(O((n^3 /
\alpha^2) \log n)\) time for any \(\alpha \in [\sqrt{n}, n]\).
We admit that our result is not likely the best running time one can achieve and that it serves more
as a first major step toward stronger results.
In particular, we are at a major disadvantage compared to the \(O(n \log n + n^2 / \alpha^2)\) time
algorithm of Chan and Rahmati~\cite{cr-iaadf-18} for discrete Fr\'echet distance in that they rely
on a constant time method for testing subsequences of points for equality and we know of no
analogous procedure for quickly testing (near) equality of subcurves.
However, it may not be the case that our own running time analysis is even tight;
perhaps a more involved analysis applied to a slight modification of our decision procedure could
lead to a better running time.
We leave open further improvements such as the one described above.

\paragraph*{Acknowledgements}
The authors would like to thank Karl Bringmann and Marvin K\"unnemann for some helpful discussions
concerning turning an approximate decision procedure into a proper approximation algorithm.


\bibliographystyle{plainurl}
\bibliography{frechet_approximation}

\end{document}